\documentclass[12pt,leqno]{article} 
\usepackage{amsmath,amstext, amsthm, amssymb, url, soul, color, hyperref}
  \numberwithin{equation}{section}
\usepackage{epsfig}
\newtheorem{thm}{Theorem}[section]

\newtheorem{cor}[thm]{Corollary}

\newcommand{\be}{\begin{equation}}
\newcommand{\ee}{\end{equation}}

\newcommand{\ba}{\begin{array}}
\newcommand{\ea}{\end{array}}

\renewcommand{\d}{\delta}
\newcommand{\bg}{\begin{gathered}}
\newcommand{\eg}{\end{gathered}}

\renewcommand{\a}{\alpha}

\renewcommand{\l}{\lambda}

\newcommand{\gauss}[2]{\genfrac{[}{]}{0pt}{}{#1}{#2}_q}

\newtheorem{rem}[thm]{Remark}
\newcommand{\bea}{\begin{eqnarray}}
\newcommand{\eea}{\end{eqnarray}}

\newcommand{\ep}{\varepsilon}

\begin{document}

\title{Asymptotics of partition functions in a fermionic matrix model  and of related $q$-polynomials}
\author{Dan Dai, Mourad E. H. Ismail and Xiang-Sheng Wang}
\date{}

\maketitle

\begin{abstract}
  In this paper, we study asymptotics of the thermal partition
  function of a model of quantum mechanical fermions
  with matrix-like index structure and quartic interactions. This
  partition function is given explicitly by a Wronskian
  of the Stieltjes-Wigert polynomials. Our asymptotic results involve
  the theta function and its derivatives.
  We also develop a new  asymptotic method for general
  $q$-polynomials.
\end{abstract}



\noindent 2010 \textit{Mathematics Subject Classification}: 33D90, 41A60.

\noindent \textit{Keywords and phrases}: partition function, matrix models, Stieltjes-Wigert polynomials, theta function, asymptotics.

\section{Introduction and statement of results}

In the past few decades, matrix models have attracted a lot of research interests due to their close relations and various applications
in many areas of mathematics and physics; for example, see \cite{Ox-handbook,Marino-book}. Quite recently, to better understand the physics of
a large number of fermionic degrees of freedom subject to non-local interactions, Anninos and Silva \cite{Ann:Sil} studied models
of quantum mechanical fermions with matrix-like index structure. Given $L,N \in \mathbb{N}$, they considered a fermionic matrix model consisting of $NL$ complex fermions $\{ \psi^{iA}, \bar{\psi}^{Ai}\}$ with $i = 1, \cdots, N$ and
$A = 1,\cdots, L$. Note that the indices $i$ and $A$ transform in the bifundamental of a $U(N) \times U(L)$ symmetry. They showed that the thermal partition function is given by
\begin{equation} \label{partition-AS}
  \widetilde{Z}_{L \times N} = \mathcal{Q} \int \prod_{i<j} \sinh^2 \left( \frac{\mu_i - \mu_j}{2} \right) \prod_{i=1}^N \cosh^L \left(\frac{\mu_i}{2}\right)   e^{-L\tilde{\gamma} \mu_i^2}
  \prod_{i=1}^N d\mu_i,
\end{equation}
where $\tilde{\gamma}>0$ is a positive parameter and the normalization constant $\mathcal{Q}$ is
\begin{equation}
  \mathcal{Q}= 2^{-L} \int \prod_{i<j} \sinh^2 \left( \frac{\mu_i - \mu_j}{2} \right) \prod_{i=1}^N e^{-L\tilde{\gamma} \mu_i^2}
  \prod_{i=1}^N d\mu_i.
\end{equation}
It is interesting to point out that the sinh term in the above integrals also appears in
the study of matrix models in Chern-Simons-matter theories; for example, see \cite{Marino-book,Tie2016}.

Later, Tierz \cite{Tie2017} realized that the partition function in \eqref{partition-AS} can be written explicitly as a Wronskian
of the Stieltjes-Wigert polynomials. Let the constant $\widehat{C}$ be given as
\begin{equation*}
  \widehat{C} =  2^{N(N-1)-NL} \exp\left( - \frac{N^3}{4L \tilde{\gamma}} - \frac{3NL}{16 \tilde{\gamma}} - \frac{N^2}{2 \tilde{\gamma}} \right),
\end{equation*}
Tierz showed that
\begin{equation} \label{eqdefZ}
  \widehat{Z}_{L \times N} = \frac{\widetilde{Z}_{L\times N}}{2^L \widehat{C}} = \frac{(-1)^{LN}}{\prod_{j=0}^{L-1} j!}
\begin{vmatrix}
S_N(\l) & S_{N+1}(\l) & \cdots & S_{N+L-1}(\l) \\
S_N^{\prime}(\l) & S_{N+1}^{\prime}(\l) & \cdots &
S_{N+L-1}^{\prime}(\l) \\
\vdots & \vdots & \vdots & \vdots \\
S_N^{(L-1)}(\l) & S_{N+1}^{(L-1)}(\l) & \cdots &
S_{N+L-1}^{(L-1)}(\l)
\end{vmatrix},
\end{equation}
where the spectral parameter is
\begin{equation}
  \l = -q^{-N-L/2}.
\end{equation}
Here $q = \exp(- \frac{1}{2 \tilde{\gamma}L}) \in (0,1)$ and $S_n(x)$ is the monic Stieltjes-Wigert polynomial
\bea
\label{eqSW}
S_n(x) = (-1)^n q^{-n^2-n/2}\sum_{k=0}^n\gauss{n}{k}  q^{k^2+k/2}
(-x)^k;
\eea
see \cite{Ismbook, Koe:Swa}. After studying the partition function $\widehat{Z}_{L \times N}$ for some finite $L$ and $N$, Tierz \cite{Tie2017} raised the question of analyzing its large $N$ limit. He discussed
briefly the case $L=1$ and pointed out a connection to the
Rogers-Ramanujan identities and possibly to their $m$ version
of Garrett, Ismail and Stanton \cite{Gar:Ism:Sta}. However, he did not
prove any asymptotic results rigorously.

In this paper, we further develop an asymptotic approach of  Wang and Wong \cite{Wang:Wong}, to find the asymptotics of $\widehat{Z}_{L\times N}$
as $N \to \infty$ for general $L\in\mathbb{N}$. Our asymptotic technique applies to general $q$-polynomials, which are not even required to be orthogonal. 
To express our results, we need the following notation for the
theta function:
\bea \label{theta-q-def}
\Theta(z):=\sum_{k=-\infty}^\infty q^{k^2}z^k;
\eea
see Whittaker and Watson \cite[p. 463]{Whit:Wat}. For convenience, let us also introduce the function related to the derivatives of the theta function as follows:
\bea \label{theta-qj-def}
\Theta_j(z):=z^j\Theta^{(j)}(z)=\sum_{k=-\infty}^\infty q^{k^2}z^k(-k)_j(-1)^j.
\eea
It is easy to see that, when  $j =0$, the above formula reduce to the theta function in \eqref{theta-q-def}; and when $j=1,$ we have
\bea
\Theta_1(z)=z\Theta'(z)=\sum_{k=-\infty}^\infty k\; q^{k^2}z^k.
\eea
Now, we are in a position  to state one of our main results in the following theorem.
\begin{thm} \label{thm-main}
  Let $q\in(0,1)$, $L,N \in \mathbb{N}$, $m=\lfloor N/2\rfloor$ and $\alpha=2m-N$. With the partition function $\widehat{Z}_{L\times N}$ defined in \eqref{eqdefZ}, we have, for all $L \in \mathbb{N}$,
\bea \label{Zln-asy}
q^{{5LN^2\over4}+{L^2N\over2}}\widehat{Z}_{L\times N}\sim{q^{ \frac{L(L-\a-1)^2}{4} }\over(q;q)_\infty^L\prod_{j=0}^{L-1}j!}\det(R), \quad \textrm{as } N \to \infty,
\eea
where $R$ is an $L\times L$ matrix with entries involving functions $\Theta_j(z)$ in \eqref{theta-qj-def} as follows
\bea
R_{ij}=\Theta_i(q^{\a-j- \frac{L-1}{2} } )=\sum_{k=-\infty}^\infty q^{k^2}q^{k \left(\a-j-\frac{L-1}{2} \right)}(-k)_i(-1)^i
\eea
for $0\le i,j\le L-1$.
\end{thm}

For simplicity, we use the symbol $\sim$ to denote asymptotic equal; namely, we write $a(N)\sim b(N)$ as $N\to\infty$ if 
$$\lim_{N\to\infty}{a(N)\over b(N)}=1.$$
When $L=1,2$, the asymptotic results in Theorem \ref{thm-main} may be put into a more concrete form.

\begin{cor}\label{cor1}
For $L =1,2$, we have
\bea \label{Zln-asy-l=1}
q^{ \frac{5N^2}{4}+ \frac{N}{2}}\widehat{Z}_{1\times N} & \sim & {q^{ \frac{\a^2}{4} } \Theta(q^{\a})\over(q;q)_\infty}
={q^{\frac{\a^2}{4} }(-q^{1+\a};q^2)_\infty(-q^{1-\a};q^2)_\infty\over(q;q^2)_\infty}
\\ \notag
& = & \begin{cases} \displaystyle{(-q;q^2)_\infty^2\over(q;q^2)_\infty}, & \textrm{if $N$ is even;} \\ \quad \\ \displaystyle{q^{\frac{1}{4} }(-1;q^2)_\infty(-q^2;q^2)_\infty\over(q;q^2)_\infty}, & \textrm{if $N$ is odd}
\end{cases}
\eea
and
\bea \label{Zln-asy-l=2}
&&q^{ \frac{5N^2}{2}+2N}\widehat{Z}_{2\times N}
\sim {q^{ \frac{(\a-1)^2}{2} }\over4}(-q;q)_\infty(-1;q)_\infty
\\&& \quad \times \left[-(q;q)_\infty^2(q;q^2)_\infty^2(q^{\a- \frac{1}{2} };q)_\infty(q^{ \frac{3}{2}-\a};q)_\infty
+(-q^{\a- \frac{1}{2} };q)_\infty(-q^{ \frac{3}{2} -\a};q)_\infty \right].\notag
\eea
\end{cor}

\begin{rem}
  In \cite{Tie2017}, Tierz conjectured that
\bea
\lim_{N\to\infty}q^{6N^2+ \frac{21N}{2} }\widehat{Z}_{1\times N}={1\over(q;q^5)_\infty^2(q^4;q^5)_\infty^2}-{1\over(q^2;q^5)_\infty^2(q^3;q^5)_\infty^2}.
\eea
Here, our asymptotic formula \eqref{Zln-asy-l=1} is different from what Tierz conjectured  in \cite{Tie2017}.
\end{rem}

\begin{rem}
  It will be of interest to evaluate the determinant of the matrix
$R$ in Theorem \ref{thm-main} for general $L$. We strongly
believe that it has a simple close form.
\end{rem}

The rest of this paper is organized as follows. In Section \ref{sec:asy-qop},
we formulate a new technique to derive the asymptotics of
$q$-polynomials. This technique is applicable to all classical
 $q$-polynomials which are orthogonal on unbounded intervals.
 We also prove an asymptotic  symmetry  property of
zeros of $q$-polynomials with positive zeros. This property states
that the product of the $k$-th largest zero and the $k$-th smallest
zero is asymptotically independent of $k$.
In the case of the Stieltjes-Wigert polynomials this property is
known, see \cite{Ism:Zhang2007,ZWang:Wong2006}. Based on the general asymptotic results in Section \ref{sec:asy-qop},
the proof of Theorem \ref{thm-main} and Corollary \ref{cor1} are done in
Section \ref{sec:proof-mainthm}. We also give another proof for the particular case $L=1$ at the end of this section.
In Section \ref{sec:asy-qop-outer} we continue the development of a new asymptotic
technique started in Section \ref{sec:asy-qop}
by considering the asymptotics in the non-oscillatory range.


\section{Asymptotics of $q$-polynomials and symmetry of zeros} \label{sec:asy-qop}

\subsection{Asymptotics of $q$-polynomials in the oscillatory interval}

To prove Theorem \ref{thm-main},  we actually  develop a new
asymptotic technique to study asymptotics of general
$q$-polynomials. Consider the following general
$q$-polynomials with real coefficients
\bea \label{pn-def}
  P_n(x)=\sum_{k=0}^nq^{k^2}f_n(k)(-x)^k,
\eea
and the related derivative functions
\bea \label{pn-j-def}
P_{n,j}(x):=x^jP_n^{(j)}(x)=\sum_{k=0}^nq^{k^2}f_n(k)(-x)^k(-k)_j(-1)^j
\eea
with $j\in \mathbb{N}.$ Define
\bea \label{X-jm-def}
X_{j,m}(x):=\sum_{k=-\infty}^\infty q^{k^2}x^k(-k-m)_j(-1)^j.
\eea
Note that this function is related to the functions $\Theta(z)$ and $\Theta_j(z)$ in \eqref{theta-q-def} and
\eqref{theta-qj-def} as $X_{0,m}(x)=\Theta(x)$ and $X_{j,0}(x)=\Theta_j(x)$.

Next, we state our asymptotic results for  general $q$-polynomials in \eqref{pn-def}.

\begin{thm}\label{thm-Pnj}
Assume that $f_n(k)$ is uniformly bounded for $n\ge0$ and $0\le k\le n$; moreover, for some $l\in(0,1)$ and $0<\d<\min(l,1-l)$,
\begin{equation}
  \sup_{n(l-\d)\le k\le n(l+\d)}|f_n(k)-1|\le\ep(n,l,\d)=o(n^{-j}), \quad \textrm{with } j \in \mathbb{N},
\end{equation}
as $n\to\infty$.
  Let $m=\lfloor nl\rfloor$, $d=\lfloor n\d \rfloor$, and $M$ be a fixed large number. Then, for the functions $P_{n,j}(x)$ given in \eqref{pn-j-def},
  we have
\bea
\qquad P_{n,j}(q^{-2m}y)=q^{-m^2}(-y)^m[X_{j,m}(-y)+O(n^j\ep(n,l,\d))+O(q^{d^2}M^dn^j)],
\eea
uniformly for $1/M\le|y|\le M$.
\end{thm}
\begin{proof}
  By a shift of variable $k\to k+m$, we have
  \bea\notag
  P_{n,j}(q^{-2m}y)=\sum_{k=-m}^{n-m}q^{k^2-m^2}f_n(k+m)(-y)^{k+m}(-k-m)_j(-1)^j=q^{-m^2}(-y)^m(I_1+I_2),
  \eea
  where
  \bea\notag
  I_1=\sum_{k=-m}^{-d}q^{k^2}f_n(k+m)(-y)^k(-k-m)_j(-1)^j+\sum_{k=d}^{n-m}q^{k^2}f_n(k+m)(-y)^k(-k-m)_j(-1)^j,
  \eea
  and
  \bea\notag
  I_2=\sum_{k=-d+1}^{d-1}q^{k^2}[f_n(k+m)-1](-y)^k(-k-m)_j(-1)^j+\sum_{k=-d+1}^{d-1}q^{k^2}(-y)^k(-k-m)_j(-1)^j.
  \eea
  Note that $0<(-k-m)_j(-1)^j<(k+m+j)^j<(d+m+j)^j(1+k-d)^j$ for all $k\ge d$.
  It then follows that
  \bea\notag
  \sum_{k=d}^\infty q^{k^2}M^k(-k-m)_j(-1)^j\le\sum_{k=0}^\infty q^{k^2+2kd+d^2}M^{k+d}(1+k)^j(d+m+j)^j=O(q^{d^2}M^dn^j),
  \eea
  which implies that $I_1=O(q^{d^2}M^dn^j)$. Furthermore, it is easily seen that $I_2= X_{j,m}(-y)+O(n^j\ep(n,l,\d))+O(q^{d^2}M^dn^j)$.
  Consequently, we obtain
\bea\notag
P_{n,j}(q^{-2m}y)=q^{-m^2}(-y)^m[X_{j,m}(-y)+O(n^j\ep(n,l,\d))+O(q^{d^2}M^dn^j)].
\eea
  This completes the proof.
\end{proof}

To illustrate the application of the above theorem, we provide asymptotics of the Stieltjes-Wigert polynomials with scaled variable $x=q^{-nt}u$, where $t\in(0,2)$. 
Let $l=t/2$ and $m=\lfloor nl\rfloor$. We obtain from \eqref{eqSW} that
\bea\notag
S_n(q^{-nt}u)={(-1)^nq^{-n^2-n/2}\over(q;q)_n}\sum_{k=0}^nq^{k^2}f_n(k)(-q^{-2m}y)^k,
\eea
where
\bea\notag
f_n(k)=(q;q)_n\gauss{n}{k},
\eea
and
$y=q^{-nt+2m+1/2}u$.
For any $0<\d<\min(l,1-l)$, we have
\bea\notag
|f_n(k)-1|&=&1-(q^{k+1};q)_{n-k}(q^{n-k+1};q)_k
\le q^{k+1}+\cdots+q^n+q^{n-k+1}+\cdots+q^n
\\\notag&\le& {q^{k+1}+q^{n-k+1}\over 1-q}
\le {1\over1-q}[q^{n(l-\d)}+q^{n(1-l-\d)}]
\eea
for all $n(l-\d)\le k\le n(l+\d)$.
It then follows from Theorem \ref{thm-Pnj} (with $j=0$) that
\bea
S_n(q^{-nt}u)=
{\Theta(-q^{-nt+2m+1/2}u)+O(q^{n(l-\d)}+q^{n(1-l-\d)})\over (-1)^n(q;q)_nq^{n^2-m^2+nmt+(n-m)/2}(-u)^{-m}}.
\eea

Similarly, Theorem \ref{thm-Pnj} also gives us asymptotic results for the $q^{-1}$-Hermite polynomials
\bea \label{q-hermite-def}
h_n(\sinh\xi)&=&\sum_{k=0}^n\gauss{n}{k}q^{k^2-nk}(-1)^ke^{(n-2k)\xi}
\notag\\&=&(-1)^nq^{n^2+n/2}e^{n\xi}S_n(q^{-n-1/2}e^{-2\xi})
\eea
and the $q$-Laguerre polynomials
\bea \label{q-lag-def}
L_n^{(\a)}(x;q)={(q^{\a+1};q)_n\over(q;q)_n}\sum_{k=0}^n\gauss{n}{k}q^{k^2+\a k}{(-x)^k\over(q^{\a+1};q)_k}.
\eea
For the $q^{-1}$-Hermite polynomials, let $\xi=-nt\ln q+\ln u$ with $t\in(-1/2,1/2)$. By choosing $f_n(k)=(q;q)_n\gauss{n}{k}$, we obtain
\bea
h_n(\sinh\xi)=
{\Theta(-q^{-n(1-2t)+2m}u^{-2})+O(q^{n(l-\d)}+q^{n(1-l-\d)})\over(-1)^m(q;q)_nq^{n^2t-m^2+nm(1-2t)}u^{2m-n}},
\eea
where $l=1/2-t$, $m=\lfloor nl\rfloor$, and $\d>0$ is any small positive number such that $\d<\min(l,1-l)$.
Regarding the $q$-Laguerre polynomials, for $t\in(0,2)$, by choosing $f_n(k)=(q^{\a+k+1};q)_{n-k}(q;q)_n\gauss{n}{k}$, we obtain
\bea
L_n^{(\a)}(q^{-nt}u;q)=
{\Theta(-q^{-nt+2m+\a}u)+O(q^{n(l-\d)}+q^{n(1-l-\d)})\over(q;q)_n^2q^{-m^2+nmt-\a m}(-u)^{-m}},
\eea
where $l=t/2$, $m=\lfloor nl\rfloor$, and $\d>0$ is any small positive number such that $\d<\min(l,1-l)$.


\subsection{Symmetry of zeros of $q$-polynomials}

It is a well-known fact that zeros of some classical $q$-orthogonal polynomials satisfy nice symmetric properties. Let $x_1<x_2<\cdots<x_n$ be the zeros of the Stieltjes-Wigert polynomial $S_n(x)$.
By \eqref{eqSW}, we have
\bea
S_n(q^{-2n-1}/x)=q^{-n^2-n/2}(-x)^{-n}S_n(x).
\eea
It is readily seen that
\bea
x_jx_{n+1-j}=q^{-2n-1}, \qquad j = 1, \cdots , n;
\eea
see also \cite[(2.8)]{Ism:Zhang2007} and \cite[(2.20)]{ZWang:Wong2006}. For the $q^{-1}$-Hermite polynomials in \eqref{q-hermite-def}, let the zeros be denoted as $\xi_1<\xi_2<\cdots\xi_n$. They also satisfy a symmetric relation as follows:
\bea
\xi_j+\xi_{n+1-j}=0.
\eea

Actually, a similar asymptotic symmetry property of polynomials zeros is satisfied for a general class of $q$-polynomials $P_n(x)$ in \eqref{pn-def},
where the coefficient $f_n(k)$ is uniformly bounded for $n\ge0$ and $0\le k\le n$. Then, for some $l\in(0,1)$ 
and $x=q^{-2nl}y$, we obtain from Theorem \ref{thm-Pnj}
\bea\label{Pn-asym}
P_n(x)\sim q^{-m^2}y^m\Theta(-q^{2(m-nl)}y),
\eea
where $m=\lfloor nl\rfloor$.
For each fixed $j=1,2,\cdots$, there exist a pair $y_j^\pm=q^{\pm(2j-1)-2(m-nl)}$ such that
$\Theta(-q^{2(m-nl)}y_j^\pm)=0$.
Consequently, for sufficiently large $n$, $P_n(x)$ has a pair of zeros $x_j^\pm\sim q^{\pm(2j-1)-2m}$; in particular, we have
\bea
x_j^+x_j^-\sim q^{-4m}.
\eea
We now apply the above result to the $q$-Laguerre polynomials in \eqref{q-lag-def} where $f_n(k)$ in \eqref{pn-def} is now given by
\bea\notag
f_n(k)={(q^{\a+1};q)_n(q;q)_n\over(q^{\a+1};q)_k}\gauss{n}{k}.
\eea
For any $l\in(0,1)$ and $0<\d<\min(l,1-l)$, we have
$$
\sup_{n(l-\d)\le k\le n(l+\d)}|f_n(k)-1|=O(q^{n(l-\d)}+q^{n(1-l-\d)}).
$$
Thus, for any fixed $l\in(0,1)$ and $j=1,2,3,\cdots,$
\bea\notag
L_n^{(\a)}(x;q)={1\over(q;q)_n^2}\sum_{k=0}^nq^{k^2}f_n(k)(-xq^{\a})^k.
\eea
has a pair of zeros $x_j^\pm\sim q^{\pm(2j-1)-2\lfloor nl\rfloor-\a}$.
This implies that for any integer $k\in(1,n)$ such that $k/n$ is bounded away from $0$ and $1$, $L_n^{(\a)}(x;q)$ has a zero
$x_k\sim q^{1-2k-\a}$; in particular, we get
\bea
x_kx_{n+1-k}\sim q^{-2n-2\a},~~n\to\infty.
\eea
Let us conduct numerical computation and choose $q=0.6,~\a=0.4$ and $n=20$. The values of $q^{2n+2\a}x_kx_{n+1-k}$ for $k=1,\cdots,10,$ are given below:
$$
0.45,0.725,0.852,0.917,0.952,0.972,0.983,0.989,0.993,0.994.
$$
We also take $q=0.5,~\a=0.7,~n=25$ and obtain the values of $q^{2n+2\a}x_kx_{n+1-k}$ for $k=1,\cdots,12,$ as follows:
$$
0.658,0.861,0.937,0.97,0.985,0.993,0.996,0.998,0.999,1.,1.,1.
$$
From the above computations, one can see that the asymptotic symmetry property is more significant with smaller $q$, larger $n$, or $k$ closer to $n/2$.

\section{Asymptotics of partition functions} \label{sec:proof-mainthm}

\subsection{Proof of Theorem \ref{thm-main}}

In this section we give a proof of Theorem \ref{thm-main} based on the general asymptotic results in Theorem \ref{thm-Pnj}.

\begin{proof}[Proof of Theorem \ref{thm-main}]
From the definition of $\widehat{Z}_{L\times N}$ in \eqref{eqdefZ}, we have
\bea\label{Z-S}
\widehat{Z}_{L\times N}=\det(S){(-1)^{LN}\over\prod_{j=0}^{L-1}j!}\prod_{j=0}^{L-1}(-1)^{N+j}q^{-(N+j)^2-(N+j)/2},
\eea
where $S$ is an $L\times L$ matrix with $ij$-th entry:
\begin{equation*}
  S_{ij}=\sum_{k_i=i}^{N+j}\gauss{N+j}{k_i}q^{k_i^2+k_i/2-(k_i-i)(N+L/2)}(-k_i)_i,
\qquad 0\le i,j\le L-1.
\end{equation*}
Similar as in the proof of Theorem \ref{thm-Pnj}, one can show that the main contribution of the sum of the right-hand side comes from the items with index $k_i$ close to $N/2$.
We may ignore the (exponentially small) items with indices $k_i<L$ or $k_i>N$, and obtain
\bea\notag
S_{ij}\sim\sum_{k_i=L}^N G_{ij}q^{k_i^2+k_i/2-(k_i-i)(N+L/2)}(-k_i)_i,
\eea
where $G_{ij}=\gauss{N+j}{k_i}$ is the $ij$-th entry of a matrix $G$.
It then follows that
\bea\label{S-G}
\qquad \det(S)\sim\sum_{L\le k_0,\cdots,k_{L-1}\le N} \det(G) q^{\sum_{i=0}^{L-1}k_i^2-k_i(N+L/2-1/2)+i(N+L/2)}\prod_{i=0}^{L-1}(-k_i)_i.
\eea
A simple calculation gives us
\bea\notag
\det(G)=\det(V)\prod_{i=0}^{L-1}{(q;q)_{N+i}\over(q;q)_{k_i}(q;q)_{N+L-1-k_i}},
\eea
where $V$ is a matrix with $ij$-th entry
\bea\notag
V_{ij}=\prod_{l=j+1}^{L-1}(1-q^{N-k_i+l}) ,
\qquad 0\le i,j\le L-1.
\eea
Here, when $j=L-1$, the empty product is understood to be 1. By row operations, the matrix $V$ can be transformed to a Vandermonde one with $ij$-th entry $q^{-jk_i}$, multiplied by certain constant factors. Indeed, one obtains
\bea\notag
\det(V)=q^{NL(L-1)/2+L(L-1)(2L-1)/6}\prod_{0\le i<j\le L-1}(q^{-k_j}-q^{-k_i}).
\eea
This kind of reduction is used systematically in Krattenthaler
\cite{Kra}.
Now, for $k_i$ near $N/2$, we have
\bea\notag
\det(G)\sim {q^{NL(L-1)/2+L(L-1)(2L-1)/6}\over(q;q)_\infty^L}\prod_{0\le i<j\le L-1}(q^{-k_j}-q^{-k_i}).
\eea
Substituting the above formula into \eqref{S-G} gives us
\bea\notag
\det(S)\sim{q^{{L(L-1)\over2}(2N+{7L-2\over6})}\over(q;q)_\infty^L}\det(T),
\eea
where $T$ is a matrix with $ij$-th entry
\begin{equation*}
  T_{ij}=\sum_{k_i=L}^N q^{k_i^2-k_i(N+L/2-1/2)-jk_i}(-k_i)_i ,
\qquad 0\le i,j\le L-1.
\end{equation*}
Let $m=\lfloor N/2\rfloor$. Shifting the index $k_i=k+m$ yields
\bea\notag
T_{ij}\sim q^{-m(N-m+L/2-1/2+j)}\sum_{k=-\infty}^\infty q^{k^2-k(N-2m+L/2-1/2+j)}(-k-m)_i.
\eea
Denote $\beta:=N-2m+(L-1)/2$. We have $\det(T)\sim q^{-mL(N-m+L-1)}\det(U)$,
where the $ij$-th entry of $U$ is given by
\bea\notag
U_{ij}=\sum_{k=-\infty}^\infty q^{k^2-k(\beta+j)}(-k-m)_i ,
\qquad 0\le i,j\le L-1.
\eea
By row operations, we obtain $\det(U)=(-1)^{L(L-1)/2}\det(R)$, where $R$ is a matrix with $ij$-th entry
\bea\notag
R_{ij}=\sum_{k=-\infty}^\infty q^{k^2-k(\beta+j)}(-k)_i(-1)^i=\Theta_i(q^{-\beta-j}) ,
\qquad 0\le i,j\le L-1.
\eea
Summarizing the above derivations, we have
\bea\notag
\widehat{Z}_{L\times N}&=&\det(S){(-1)^{L(L-1)\over2}q^{-LN^2-{LN\over2}-{L(L-1)\over2}(2N+{4L+1\over6})}\over\prod_{j=0}^{L-1}j!}
\\\notag&\sim&\det(R){q^{-LN^2-{LN\over2}+{L(L-1)^2\over4}-mL(N-m+L-1)}\over(q;q)_\infty^L\prod_{j=0}^{L-1}j!},
\eea
where a rigorous error estimation for more general situations has been given as in Theorem \ref{thm-Pnj}.
A further simplification gives us the result in \eqref{Zln-asy}.
\end{proof}

\subsection{Proof of \eqref{Zln-asy-l=2}}

A direct application of \eqref{Zln-asy} with $L=2$ gives
\bea \label{Zln-asy-l=2'}
q^{5N^2/2+2N}\widehat{Z}_{2\times N}
\sim \frac{\Theta(q^{\alpha-1/2})\Theta'(q^{\alpha-3/2})-
q\Theta'(q^{\alpha-1/2})\Theta(q^{\alpha-3/2})}{ q^{1-\alpha^2/2}(q;q)_\infty^2}.
\eea
Denote $q=e^{\pi i\tau}$ and $a=\pi\tau(\a/2-3/4)$.
Making use of the theta functions defined in Whittaker and Watson \cite[Chapter 21]{Whit:Wat}, we may rewrite
\bea
\Theta(q^{\alpha-3/2})&=&\sum_{k=-\infty}^\infty q^{k^2}e^{2kia}=\vartheta_3(a;q),\notag\\
\Theta(q^{\alpha-1/2})&=&\vartheta_3(a+\pi\tau/2;q)=q^{(1-\a)/2}\vartheta_2(a;q),\notag\\
\Theta'(q^{\alpha-3/2})&=&{q^{3/2-\a}\over 2i}\vartheta_3'(a;q),\notag\\
\Theta'(q^{\alpha-1/2})&=&{q^{1-3\a/2}\over 2i}[\vartheta_2'(a;q)-i\vartheta_2(a;q)].\notag
\eea
It then follows from \cite[Chapter 21]{Whit:Wat} and  \cite[Chapter 20]{DLMF} that
\bea\label{Th-Wronskian}
&&\Theta(q^{\alpha-1/2})\Theta'(q^{\alpha-3/2})-
q\Theta'(q^{\alpha-1/2})\Theta(q^{\alpha-3/2})
\\&=&{q^{2-3\a/2}\over 2i}[\vartheta_2(a;q)\vartheta_3'(a;q)-\vartheta_3(a;q)\vartheta_2'(a;q)+i\vartheta_3(a;q)\vartheta_2(a;q)]
\notag\\&=&{q^{2-3\a/2}\over 2i}[\vartheta_4(0;q)^2\vartheta_4(a;q)\vartheta_1(a;q)+i\vartheta_3(a;q)\vartheta_2(a;q)]
\notag\\&=&{q^{2-3\a/2}\vartheta_2(0;\sqrt q)\over 4i}[\vartheta_4(0;q)^2\vartheta_1(a;\sqrt q)+i\vartheta_2(a;\sqrt q)].
\notag
\eea
Recalling the Jacobi triple product identity (see \cite{And:Ask:Roy} or \cite[II.33]{Gas:Rah})
\bea \label{F-def}
\sum_{k=-\infty}^\infty q^{k^2/2} z^k = (q;q)_\infty(-z\sqrt{q};q)_\infty
(-\sqrt{q}/z;q)_\infty,
\eea
we have
\bea
\vartheta_1(a;\sqrt q)&=&-iq^{\a/2-5/8}(q;q)_\infty(q^{\a-1/2};q)_\infty(q^{3/2-\a};q)_\infty,\notag\\
\vartheta_2(a;\sqrt q)&=&q^{\a/2-5/8}(q;q)_\infty(-q^{\a-1/2};q)_\infty(-q^{3/2-\a};q)_\infty,\notag\\
\vartheta_4(0;q)&=&(q;q)_\infty(q;q^2)_\infty,\notag\\
\vartheta_2(0;\sqrt q)&=&q^{1/8}(q;q)_\infty(-q;q)_\infty(-1;q)_\infty.\notag
\eea
Substituting the above formulas into \eqref{Th-Wronskian} yields
\bea
&&\Theta(q^{\alpha-1/2})\Theta'(q^{\alpha-3/2})-
q\Theta'(q^{\alpha-1/2})\Theta(q^{\alpha-3/2})
={q^{3/2-\a}\over4}(q;q)_\infty^2(-q;q)_\infty(-1;q)_\infty
\notag\\&&\times[-(q;q)_\infty^2(q;q^2)_\infty^2(q^{\a-1/2};q)_\infty(q^{3/2-\a};q)_\infty
+(-q^{\a-1/2};q)_\infty(-q^{3/2-\a};q)_\infty],\notag
\eea
which, together with \eqref{Zln-asy-l=2'} gives \eqref{Zln-asy-l=2}.

\subsection{Another proof for the case $L=1$}

For the case $L=1$, \eqref{Zln-asy-l=1} is a simple application of \eqref{Zln-asy} and \eqref{F-def}.
However, there is another proof due to the simple structure of $\widehat{Z}_{1 \times N}$. This new proof is based on the integral representation of the partition $\widehat{Z}_{1 \times N}$.
For brevity, let us denote the summation in \eqref{F-def} by
\bea\notag
F(z):=\sum_{k=-\infty}^\infty q^{k^2/2} z^k = (q;q)_\infty(-z\sqrt{q};q)_\infty
(-\sqrt{q}/z;q)_\infty.
\eea
Cauchy's residue theorem gives us
\bea\notag
q^{k^2/2}={1\over2\pi i}\oint_C F(z){dz\over z^{k+1}},
\eea
where $C$ is a positively oriented contour around the origin. It then follows from the definition of $\widehat{Z}_{L\times N}$ in \eqref{eqdefZ} that
\bea\notag
q^{N^2+N/2}\widehat{Z}_{1\times N}=\sum_{k=0}^N \gauss{N}{k}q^{k^2-kN}
={1\over2\pi i}\oint_C\sum_{k=0}^N \gauss{N}{k}q^{k^2/2-kN}F(z){dz\over z^{k+1}}.
\eea
Using the $q$-binomial theorem
\bea
\label{eqqbinom}
(x;q)_n = \sum_{k=0}^n \gauss{n}{k} q^{\binom{k}{2}} (-x)^k,
\eea
(cf.  \cite{And:Ask:Roy}, \cite[(II.4)]{Gas:Rah}), we have
\bea
q^{N^2+N/2}\widehat{Z}_{1\times N}={1\over2\pi i}\oint_C F(z)(-q^{-N+1/2}/z;q)_N{dz\over z}.
\eea
Let us introduce a change of variable $z=q^{-m}u$ with $m=\lfloor N/2\rfloor$.
Since
\bea \label{F-qm-u}
F(q^{-m}u)=\sum_{k=-\infty}^\infty q^{k^2/2-km} u^k=F(u)q^{-m^2/2}u^m,
\eea
and
\bea\notag
(-q^{-N+1/2+m}/u;q)_N
&=&(-q^{-N+1/2+m}/u;q)_m(-q^{-N+1/2+2m}/u;q)_{N-m}
\\\notag&=&q^{-Nm+3m^2/2}u^{-m}(-q^{N+1/2-2m}u;q)_m(-q^{-N+1/2+2m}/u;q)_{N-m},
\eea
we obtain
\bea\notag
q^{N^2+N/2}\widehat{Z}_{1\times N}
&=&{q^{-Nm+m^2}\over2\pi i}\oint_C F(u)(-q^{N+1/2-2m}u;q)_m(-q^{-N+1/2+2m}/u;q)_{N-m}{du\over u}
\\\notag&\sim&{q^{-Nm+m^2}\over2\pi i}\oint_C F(u)(-q^{N+1/2-2m}u;q)_\infty(-q^{-N+1/2+2m}/u;q)_\infty{du\over u}.
\eea
Note that
\bea\notag
F(q^{N-2m}u)=(q;q)_\infty(-q^{N+1/2-2m}u;q)_\infty(-q^{-N+1/2+2m}/u;q)_\infty,
\eea
it then follows that
\bea\notag
(q;q)_\infty q^{N^2+N/2+Nm-m^2}\widehat{Z}_{1\times N}\sim{1\over2\pi i}\oint_C F(u)F(q^{N-2m}u){du\over u}.
\eea
The right-hand side equals to the constant term of the product $F(u)F(q^{N-2m}u)$:
\bea\notag
\sum_{k=-\infty}^\infty q^{(-k)^2/2}q^{k^2/2}q^{(N-2m)k}
=\Theta(q^{N-2m})=\Theta(q^{2m-N}).
\eea
Thus, the above two formulas give us
\bea
\label{Tierz}
q^{5N^2/4+N/2}\widehat{Z}_{1\times N}\sim{q^{\alpha^2/4}\Theta(q^\alpha)\over(q;q)_\infty},
\eea
where $\alpha=2m-N$. Moreover, if $N=2m$ is even, we have
$$\Theta(q^{2m-N})=(q^2;q^2)_\infty(-q;q^2)_\infty^2,$$
and if $N=2m+1$ is odd
$$\Theta(q^{2m-N})=(q^2;q^2)_\infty(-1;q^2)_\infty(-q^2;q^2)_\infty.$$
Then, \eqref{Zln-asy-l=1} immediately follows from a combination of the above three formulas.

\section{Asymptotics of $q$-polynomials in the non-oscillatory interval} \label{sec:asy-qop-outer}


Theorem \ref{thm-Pnj} gives asymptotics of $P_{n,j}(x)$ in the oscillatory interval:
\bea
{\ln x\over n\ln q}\in(-2,0).
\eea
To make our asymptotic technique complete, we further provide asymptotics in the interval where $P_{n,j}(x)$ is non-oscillatory. To this end, we shall introduce the generalized theta function:
\bea \label{def-Phi}
\Phi(z):=\sum_{k=0}^\infty a_kq^{k^2}z^k
\eea
and the associated functions:
\bea \label{def-Phi-j}
\Phi_j(z):=z^j\Phi^{(j)}(z)=\sum_{k=0}^\infty a_kq^{k^2}z^k(-k)_j(-1)^j.
\eea
Here, $\{a_k\}_{k\ge0}$ is a sequence with uniform bound.
When $a_k\equiv1$, we have $\Phi(z)+\Phi(1/z)=\Theta(z)+1$, where $\Theta(z)$ is given in \eqref{theta-q-def}. It is worth pointing out that, when $a_k=(-1)^k/(q;q)_k$,
the function $\Phi(z)$ is the same as the Ramanujan function (i.e., the $q$-Airy function); see \cite{Ism:Zhang2007}.

\begin{thm}\label{thm-outer}
Assume that $f_n(k)$ is uniformly bounded for $n\ge0$ and $0\le k\le n$, and there exist uniformly bounded sequence $\{a_k\}_{k\ge0}$ and $\d\in(0,1)$ such that
\begin{equation}
  \sup_{0\le k\le n\d}|f_n(k)-a_k|\le\ep(n,\d)=o(1)
\end{equation}
as $n\to\infty$.
  Let $d=\lfloor n\d \rfloor$ and $M$ be a fixed large number. Then, for the functions $P_{n,j}(x)$ given in \eqref{pn-j-def},
  we have
\bea
P_{n,j}(q^{nt}y)=\Phi_j(-q^{nt}y)+O(\ep(n,\d))+O(q^{d^2}M^d),
\eea
uniformly for $|y|\le M$ and $t\ge0$.
\end{thm}
\begin{proof}
We split the error term into two sums:
\bea\notag
P_{n,j}(q^{nt}y)-\Phi_j(-q^{nt}y)=I_1+I_2,
\eea
where
  \bea\notag
  I_1=\sum_{k=d}^nq^{k^2}f_n(k)(-q^{nt}y)^k(-k)_j(-1)^j-\sum_{k=d}^\infty a_kq^{k^2}(-q^{nt}y)^k(-k)_j(-1)^j=O(q^{d^2}M^d),
  \eea
  and
  \bea\notag
  I_2=\sum_{k=0}^{d-1}q^{k^2}[f_n(k)-a_k](-q^{nt}y)^k(-k)_j(-1)^j=O(\ep(n,\d)).
  \eea
Here, we have used the estimation:
\bea\notag
\sum_{k=d}^\infty q^{k^2}M^k(-k)_j(-1)^j&=&\sum_{k=0}^\infty q^{k^2+2kd+d^2}M^{k+d}(-k-d)_j(-1)^j
\\&\le& q^{d^2}M^d\sum_{k=0}^\infty q^{k^2}q^{2kd}(k+d+j)^j=O(q^{d^2}M^d)
\eea
due to uniform boundedness of $q^{2kd}(k+d+j)^j$ for $k\ge0$ and $d\ge0$.
This completes the proof.
\end{proof}

To conduct asymptotic analysis of $P_{n,j}(q^{nt}y)$ with $t\le-2$, we set $t=-(2+s)$ with $s\ge0$ and change the index $k$ to $n-k$.
It follows that
\bea\notag
P_{n,j}(q^{nt}y)=(-q^{n+ns}/y)^{-n}\sum_{k=0}^nq^{k^2}f_n(n-k)(-q^{ns}/y)^k(-n+k)_j(-1)^j.
\eea
We shall make use of the following special function
\bea
\Psi_{j,n}(z)=\sum_{k=0}^n a_kq^{k^2}z^k(-n+k)_j(-1)^j.
\eea
Note that $\Psi_{0,n}(z)=\Phi_0(z)=\Phi(z)$ with $\Phi(z)$ defined in \eqref{def-Phi}.
Moreover, when $a_k\equiv1$, we have the relation $\Psi_{j,n}(z)=(-1)^jX_{j,j-n-1}(z)$ with $X_{j,m}(z)$ given in \eqref{X-jm-def}.
\begin{thm}\label{thm-outer-2}
Assume that $f_n(k)$ is uniformly bounded for $n\ge0$ and $0\le k\le n$, and there exist uniformly bounded sequence $\{a_k\}_{k\ge0}$ and $\d\in(0,1)$ such that
\begin{equation}
  \sup_{0\le k\le n\d}|f_n(n-k)-a_k|\le\ep(n,\d)=o(n^{-j})
\end{equation}
as $n\to\infty$.
  Let $d=\lfloor n\d \rfloor$ and $M$ be a fixed large number.
  Then, for the functions $P_{n,j}(x)$ given in \eqref{pn-j-def},
  we have
\bea
\qquad \quad P_{n,j}(q^{nt}y)=(-q^{n+nt}y)^n[\Psi_{j,n}(-q^{-2n-nt}/y)+O(n^j\ep(n,\d))+O(q^{d^2}M^dn^j)],
\eea
uniformly for $|y|\ge 1/M$ and $t\le-2$.
\end{thm}

\begin{proof}
Again, we split the error term into two sums:
\bea\notag
(-q^{n+ns}/y)^nP_{n,j}(q^{nt}y)-\Phi_j(-q^{ns}/y)=I_1+I_2,
\eea
where
  \bea\notag
  I_1=\sum_{k=d}^nq^{k^2}f_n(n-k)(-q^{ns}/y)^k(-n+k)_j(-1)^j-\sum_{k=d}^\infty a_kq^{k^2}(-q^{ns}/y)^k(-n+k)_j(-1)^j,
  \eea
  and
  \bea\notag
  I_2=\sum_{k=0}^{d-1}q^{k^2}[f_n(n-k)-a_k](-q^{ns}/y)^k(-n+k)_j(-1)^j=O(n^j\ep(n,\d)).
  \eea
  Since $|(-n+k)_j|\le (n+d)^j(1+k-d)^j$ for all $k\ge d$, we have
  \bea\notag
  \sum_{k=d}^\infty q^{k^2}M^k|(-n+k)_j|
  \le \sum_{k=0}^\infty q^{k^2+2kd+d^2}M^{k+d}(1+k)^j(n+d)^j=O(q^{d^2}M^dn^j).
  \eea
Thus, $I_1=O(q^{d^2}M^dn^j)$.
This completes the proof.
\end{proof}

\section*{Acknowledgements}

D.D was partially supported by a grant from the City
University of Hong Kong (Project No. 7004864), and grants from the Research Grants Council of the Hong Kong Special Administrative Region, China (Project No. CityU 11300115, CityU 11303016).


\noindent D.D,  City University of Hong Kong, Tat Chee Avenue, Kowloon Tong, Hong Kong \\
email: dandai@cityu.edu.hk

\medskip

\noindent M.E.H.I,
University of Central Florida, Orlando, FL 32816, USA \\
  email: mourad.eh.ismail@gmail.com

  \medskip

  \noindent X.S.W,  University of Louisiana at Lafayette, Lafayette, LA 70503, USA     \\
  email: xswang@louisiana.edu

\end{document}